\documentclass[a4paper,fleqn]{cas-sc}

\usepackage[authoryear,longnamesfirst]{natbib}
\usepackage{graphicx}%
\usepackage{multirow}%
\usepackage{amsmath,amssymb,amsfonts}%
\usepackage{amsthm}%
\usepackage{mathtools}%
\usepackage{mathrsfs}%
\usepackage[title]{appendix}%
\usepackage{xcolor}%
\usepackage{textcomp}%
\usepackage{manyfoot}%
\usepackage{booktabs}%
\usepackage{tabularx}%
\usepackage{array}%
\usepackage{algorithm}%
\usepackage{algorithmicx}%
\usepackage{algpseudocode}%
\usepackage{enumitem}%
\usepackage{url}
\usepackage{capt-of}

\newcommand{\prox}{\operatorname{prox}}
\newcommand{\dist}{\operatorname{dist}}
\newcommand{\sign}{\operatorname{sign}}
\newcommand{\ind}{\mathbf{1}}
\newcommand{\eps}{\varepsilon}

\theoremstyle{plain}
\newtheorem{theorem}{Theorem}
\newtheorem{proposition}{Proposition}
\newtheorem{lemma}{Lemma}
\theoremstyle{definition}
\newtheorem{assumption}{Assumption}
\theoremstyle{remark}
\newtheorem{remark}{Remark}

\def\tsc#1{\csdef{#1}{\textsc{\lowercase{#1}}\xspace}}
\tsc{WGM}
\tsc{QE}
\tsc{EP}
\tsc{PMS}
\tsc{BEC}
\tsc{DE}

\begin{document}
\let\WriteBookmarks\relax
\def\floatpagepagefraction{1}
\def\textpagefraction{.001}
\shorttitle{Distributed Stochastic Smoothing ADMM}

\title [mode = title]{Distributed Stochastic Smoothing ADMM for Penalized Quantile Regression}                      
\tnotemark[1]

\tnotetext[1]{This work was supported by the Research Start-up Fund
for Talents of the School of Mathematics and Statistics,
Yunnan University (Grant No. CZ305010126XX).}

\author[1]{Rongmei Liang}
\ead{liang_r_m@163.com}
\credit{conceptualization, supervision, methodology, writing--review and editing}

\affiliation[1]{organization={Department of Statistics and Data Science, Southern University of Science and Technology},
                addressline={1088 Xueyuan Avenue, Nanshan District}, 
                city={Shenzhen},
                postcode={518055}, 
                state={Guangdong},
                country={China}}

\author[2,3]{Xiaofei Wu}
\cormark[1]
\ead{xfwu1016@ynu.edu.cn}
\credit{methodology, software, formal analysis, visualization, writing--original draft, writing--review and editing}

\affiliation[2]{organization={Yunnan Key Laboratory of Statistical Modeling and Data Analysis, Yunnan University},
                addressline={East Outer Ring South Road, University Town}, 
                city={Kunming},
                postcode={650504}, 
		state={Yunnan},
                country={China}}

\affiliation[3]{organization={School of Mathematics and Statistics, Yunnan University},
                addressline={East Outer Ring South Road, University Town}, 
                city={Kunming},
                postcode={650504}, 
                state={Yunnan}, 
                country={China}}

\cortext[cor1]{Corresponding author}

\begin{abstract}
Quantile regression is well suited to heterogeneous and heavy-tailed data, but computation becomes challenging for large, distributed data sets because the check loss is nonsmooth. We propose a distributed stochastic smoothing alternating direction method of multipliers (DSS-ADMM) for horizontally partitioned penalized quantile regression. Each worker computes a mini-batch gradient of a Huber-smoothed check loss, and a coordinator performs a single proximal aggregation step for the regularizer. Raw observations remain local, worker updates run in parallel, and the method requires no matrix inversion. For proper, closed, and convex penalties, stacking the local coefficient vectors yields a standard two-block stochastic ADMM formulation. With fixed smoothing, we establish an expected $O(\log K/\sqrt K)$ joint objective-feasibility bound and an explicit $\eps/4$ approximation term for the original check-loss objective; when the smooth block is strongly convex, the bound improves to $O(\log K/K)$. We also characterize the scope of an extension to the minimax concave penalty and the smoothly clipped absolute deviation penalty. Reproducible simulations consider both homogeneous worker partitions, in which observations are independently and identically distributed across workers, and heterogeneous partitions, in which worker-specific covariate distributions differ. Sensitivity studies and analyses of the diabetes and Engel data illustrate the trade-offs among per-observation gradient evaluations, communication, consensus, sparsity, and prediction.
\end{abstract}



\begin{keywords}
Quantile regression \sep stochastic ADMM \sep distributed optimization \sep weakly convex regularization
\end{keywords}

\maketitle

\section{Introduction}\label{sec:introduction}
Quantile regression estimates conditional distributional features rather than only a conditional mean \citep{koenker1978,koenker2005}. It is therefore useful when errors are asymmetric, heteroscedastic, or heavy-tailed. Modern applications also require sparse or structured coefficients, motivating the least absolute shrinkage and selection operator (lasso) \citep{tibshirani1996}, elastic net \citep{zouhastie2005}, group lasso \citep{yuanlin2006}, and high-dimensional quantile regularization \citep{belloni2011,wuliu2009}. The resulting objective is typically nonsmooth in both its data-fitting and regularization terms.

Horizontal partitioning creates a second difficulty. Observations may be stored at hospitals, laboratories, sensors, or administrative sites and cannot be pooled. The alternating direction method of multipliers (ADMM) separates local fitting from global regularization through consensus constraints \citep{boyd2011}, but exact local minimization repeatedly scans all observations and may require numerical inner iterations. A stochastic local step can reduce this cost, provided that its random gradient estimator and its interaction with the consensus constraint are analyzed explicitly.

Online and stochastic ADMM methods provide the relevant optimization foundation. \citet{wangbanerjee2012} studied online ADMM; \citet{suzuki2013} combined online ADMM with dual averaging and proximal-gradient ideas; \citet{ouyang2013} established stochastic objective-feasibility bounds; \citet{azadi2014} refined stochastic ADMM complexity; and \citet{zhong2014} developed a faster variance-controlled variant. These analyses differ in update order and assumptions, but the stochastic block is smooth in their standard formulations. The raw check loss is not differentiable at zero, so those results cannot be invoked without modifying the loss or the stochastic oracle.

ADMM has become an important tool for large-scale penalized quantile regression. \citet{yulin2017} gave a distributable ADMM formulation for big data. \citet{yu2017} proposed the single-loop QPADM algorithm for nonconvex penalization. \citet{fan2021} used a slack-variable representation for distributed big data, and \citet{wang2024} combined random features with iterative ADMM for communication-efficient nonparametric quantile regression. Related parallel developments include feature splitting, consensus regularization, and Gaussian back substitution \citep{wu2025a,wu2025b,wu2025c}. These methods clarify how ADMM can exploit statistical structure, but they do not by themselves yield a mini-batch stochastic-gradient theorem for the smoothed consensus model studied here.

Smoothing of the quantile objective predates recent smoothing-ADMM algorithms. The piecewise quadratic-linear function used below is, up to an additive constant, the classical Huber loss \citep{huber1964}; finite and differentiable approximations of the quantile objective were studied earlier by \citet{chen2007} and \citet{zheng2011}. \citet{mirzaeifard2024} subsequently used this Huber-type surrogate in a single-loop smoothing ADMM for sparse quantile regression, and federated and decentralized extensions followed \citep{mirzaeifard2025,mirzaeifard2025dsad}. Their results motivate our surrogate, but the present randomness is different: every worker samples a local mini-batch, and the resulting martingale and conditional-variance terms enter the convergence proof.

A recent development is particularly relevant to the stochastic component of this work. \citet{long2026} proposed stochastic variance-reduced recursive momentum ADMM (SVRRM-ADMM), based on a loopless stochastic variance-reduced gradient estimator, together with an accelerated variant (ASVRRM-ADMM). For a global finite-sum linearly constrained problem with a smooth, possibly nonconvex first block and a convex, possibly nonsmooth second block, they obtained an $O(1/T)$ stationarity rate without imposing a uniform bounded-variance assumption. Under the Kurdyka--{\L}ojasiewicz (KL) property, they also established finite expected length and rates determined by the KL exponent. Because smoothing expresses each local loss in our model as a finite sum of differentiable terms, their recursive estimator provides a possible route to reducing stochastic-gradient error at the workers. Their theorem, however, does not directly strengthen the distributed objective-feasibility bound established here. Such an extension would require a distributed Lyapunov analysis that incorporates worker-specific estimator memories, probabilistic snapshot refreshes, the consensus operator, and the coordinator proximal step. Similarly, an accelerated distributed rate would require simultaneous control of extrapolation, gradient-tracking error, and consensus.

The paper makes four contributions.
\begin{enumerate}
\item We propose an inversion-free DSS-ADMM algorithm for horizontally partitioned penalized quantile regression. Workers access only local mini-batches, while the coordinator performs a single regularizer-dependent proximal update.
\item For convex proximable penalties, we stack all local models into one first block and thereby avoid an uncontrolled multi-block interpretation. We establish an expected $O(\log K/\sqrt K)$ joint objective-feasibility bound for fixed smoothing, an $\eps/4$ transfer bound for the original check loss, and an $O(\log K/K)$ result when the smooth block is strongly convex.
\item We characterize which results extend to the minimax concave penalty (MCP) and the smoothly clipped absolute deviation (SCAD) penalty. Their coordinator update is single-valued when the aggregate ADMM curvature exceeds the weak-convexity modulus. Convergence to a Karush--Kuhn--Tucker (KKT) point, however, is conditional on vanishing stochastic error and cannot in general be guaranteed for a fixed mini-batch with a nonvanishing noise level.
\item We report simulations under homogeneous and heterogeneous worker partitions, sensitivity analyses for the worker count and smoothing parameter, robustness and dimensional-scaling experiments, a weakly convex experiment, and analyses of the diabetes and Engel data. The comparisons distinguish prediction, sparsity, consensus, per-observation gradient evaluations, communication rounds, and serial emulation time.
\end{enumerate}

Three distinctions define the theoretical scope. First, stochasticity is at the observation level within each site, not merely at the worker-participation level. Conditional on the history, every local estimator is unbiased for its weighted smooth gradient. This permits the worker vectors to be stacked into one random first block and makes the dependence on local batch sizes explicit through a conditional-variance bound. Second, the main theorem uses fixed smoothing. The smoothing parameter controls both approximation and conditioning: the uniform approximation error is proportional to $\eps$, whereas the gradient Lipschitz constant grows as $1/\eps$. A single-loop schedule with $\eps_k\downarrow0$ changes both the objective and the admissible step size at every iteration and therefore needs a separate time-varying analysis. Third, convex and weakly convex penalties are separated. Convex subgradient inequalities and Jensen's inequality support the ergodic objective-feasibility analysis, whereas MCP and SCAD require stationarity arguments, boundedness, and vanishing stochastic error.

For clarity, the main notation is summarized here before the formal model is introduced. The number of workers is $M$, worker $m$ stores $n_m$ observations, the total sample size is $N=\sum_{m=1}^M n_m$, and the worker weight is $\alpha_m=n_m/N$. The common regression dimension is $p$, the quantile level is $\tau\in(0,1)$, and the smoothing parameter is $\eps>0$. At communication round $k\in\{0,\ldots,K-1\}$, $\mathcal B_m^k$ denotes the local mini-batch, $b_{m,k}=|\mathcal B_m^k|$ its size, and $g_m^k\in\mathbb{R}^p$ the local stochastic-gradient estimator. The local coefficient is $\theta_m^k\in\mathbb{R}^p$, the coordinator consensus vector is $z^k\in\mathbb{R}^p$, and $u_m^k\in\mathbb{R}^p$ is the scaled dual variable. The ADMM quadratic penalty is $\varrho>0$, the local proximal step is $\eta_{k+1}>0$, and $\prox$ denotes the proximal map. In the stacked formulation, $\beta=((\theta_1)^\mathsf{T},\ldots,(\theta_M)^\mathsf{T})^\mathsf{T}\in\mathbb{R}^{Mp}$, $A_1=I_{Mp}$, $A_2=-(\mathbf 1_M\otimes I_p)$, and $A_1\beta+A_2z=0$ represents all consensus equations. The functions $f_{1,\eps}$ and $P$ are respectively the smooth data-fitting block and the proximable regularization block. The symbol $\mathbb{E}$ denotes expectation, $\dist$ denotes Euclidean distance to a set, and a superscript $\star$ denotes an optimal solution or optimal value.

The distributed formulation allows unequal site sizes and non-identically distributed covariates, although the target remains one common coefficient vector. Full worker participation is assumed in the theorem. Partial participation, stale messages, compression, privacy noise, and vertical feature splitting are outside the proved scope. The method is federated only in the computational sense that raw records remain local; it is not by itself a formal privacy mechanism.

The remainder of the paper is organized as follows. Section~\ref{sec:model} introduces the model and smoothing approximation. Section~\ref{sec:algorithm} presents DSS-ADMM, Section~\ref{sec:convex-theory} develops the convex theory and discusses recursive-momentum variance reduction and acceleration, and Section~\ref{sec:weakly-convex} delineates the weakly convex extension. Section~\ref{sec:experiments} reports numerical results, and Section~\ref{sec:conclusion} concludes. Appendix~\ref{app:convex-proofs} contains the convex proofs, whereas Appendix~\ref{app:weakly-convex} gives the conditional weakly convex argument.

\section{Distributed penalized quantile regression}\label{sec:model}
\subsection{Model and consensus representation}
Let $M$ workers hold disjoint samples
\[
\mathcal D_m=\{(x_{mi},y_{mi}):i=1,\ldots,n_m\},
\qquad N=\sum_{m=1}^M n_m,
\]
where $x_{mi}\in\mathbb{R}^p$ and $\alpha_m=n_m/N$ is the proportion of observations stored at worker $m$. For a quantile level $\tau\in(0,1)$, the check loss applied to a residual $u$ is
\begin{align*}
\rho_\tau(u)&=u\{\tau-\ind(u<0)\}=\frac12\{|u|+(2\tau-1)u\}.
\end{align*}
Define the local empirical loss
\[
f_m(\theta)=\frac1{n_m}\sum_{i=1}^{n_m}
\rho_\tau(y_{mi}-x_{mi}^{\mathsf T}\theta).
\]
Weighting each local loss by $\alpha_m$ recovers the empirical loss over all $N$ observations without moving any record between workers. The resulting penalized quantile-regression problem is
\begin{equation}
\begin{aligned}
\min_{\theta\in\mathbb{R}^p}F(\theta)
=\sum_{m=1}^M\alpha_m f_m(\theta)+P(\theta),
\end{aligned}
\label{eq:target}
\end{equation}
where $P$ is first assumed proper, closed, convex, and proximable. Examples include the lasso, elastic net, group lasso, fused lasso, and indicator functions of convex constraints \citep{parikh2014}.

Problem~\eqref{eq:target} couples all sites through the common coefficient $\theta$. To expose the distributed structure, introduce a local copy $\theta_m$ at each worker and a coordinator variable $z$. Enforcing $\theta_m=z$ gives the equivalent consensus formulation
\begin{align}
\min_{\theta_1,\ldots,\theta_M,z}\quad &
\sum_{m=1}^M\alpha_m f_m(\theta_m)+P(z),\notag\\
\text{s.t.}\quad &\theta_m=z,\qquad m=1,\ldots,M.
\label{eq:consensus}
\end{align}
The loss terms in \eqref{eq:consensus} can now be evaluated in parallel from local data, whereas the regularizer is handled once through $z$. The remaining obstacle is that $\rho_\tau$ is nondifferentiable at zero, so a direct stochastic-gradient update is not available. We therefore introduce a smooth approximation before deriving the distributed updates.

\subsection{Smoothing the check loss}
For $\eps>0$, replace the absolute-value term in the check loss by the quadratic-linear function
\begin{equation}
\begin{aligned}
 h_\eps(u)&=
 \begin{cases}
 |u|,&|u|\ge\eps,\\[1mm]
 u^2/(2\eps)+\eps/2,&|u|<\eps,
 \end{cases}\\
\rho_{\tau,\eps}(u)&=\frac12\{h_\eps(u)+(2\tau-1)u\}.
\end{aligned}
\label{eq:smooth}
\end{equation}
The surrogate agrees with $|u|$ outside the interval $(-\eps,\eps)$ and replaces its kink at zero by a quadratic segment. Up to an additive constant, $h_\eps$ is the Huber loss \citep{huber1964}. Smooth quantile approximations appeared before the SIAD family \citep{chen2007,zheng2011}; the later SIAD papers used this surrogate in centralized, federated, and decentralized ADMM \citep{mirzaeifard2024,mirzaeifard2025,mirzaeifard2025dsad}.

The corresponding derivative, which supplies the scalar factor in each sample gradient, is
\begin{equation}
\begin{aligned}
\psi_{\tau,\eps}(u)=\rho'_{\tau,\eps}(u)=\frac12
\begin{cases}
\sign(u)+2\tau-1,&|u|\ge\eps,\\
u/\eps+2\tau-1,&|u|<\eps.
\end{cases}
\end{aligned}
\label{eq:psi}
\end{equation}
\begin{lemma}[Smoothing properties]\label{lem:smoothing-properties}
For every $u\in\mathbb{R}$,
$0\le\rho_{\tau,\eps}(u)-\rho_\tau(u)\le\eps/4$,
$|\psi_{\tau,\eps}(u)|\le\max(\tau,1-\tau)$, and
$\psi_{\tau,\eps}$ is $1/(2\eps)$-Lipschitz.
\end{lemma}
Lemma~\ref{lem:smoothing-properties} records the two features needed below. The first inequality controls the bias introduced by smoothing, while boundedness and Lipschitz continuity of $\psi_{\tau,\eps}$ control the stochastic-gradient magnitude and the smoothness constant. Accordingly, define the smoothed local empirical loss
\[
f_{m,\eps}(\theta)=\frac1{n_m}\sum_{i=1}^{n_m}
\rho_{\tau,\eps}(y_{mi}-x_{mi}^{\mathsf{T}}\theta).
\]
Applying the chain rule to each residual gives
\begin{equation}
\nabla f_{m,\eps}(\theta)=-\frac1{n_m}\sum_{i=1}^{n_m}
x_{mi}\psi_{\tau,\eps}(y_{mi}-x_{mi}^{\mathsf{T}}\theta),
\label{eq:fullgrad}
\end{equation}
Let $X_m$ denote the $n_m\times p$ design matrix whose $i$th row is $x_{mi}^{\mathsf T}$. Because $\psi_{\tau,\eps}$ is $1/(2\eps)$-Lipschitz, the gradient in \eqref{eq:fullgrad} is Lipschitz with constant bounded by
\begin{equation}
L_{m,\eps}\le \frac{\|X_m\|_2^2}{2\eps n_m}.
\label{eq:Lm}
\end{equation}
Thus smoothing converts the nonsmooth local loss into a differentiable finite sum with an explicit smoothness bound. This structure permits the mini-batch estimator introduced in Section~\ref{sec:algorithm}, while the approximation inequality in Lemma~\ref{lem:smoothing-properties} later transfers the convergence result back to the original check-loss objective.

\section{The DSS-ADMM algorithm}\label{sec:algorithm}
\subsection{Local stochastic gradients}
At round $k$, worker $m$ samples a mini-batch $\mathcal B_m^k$ of size $b_{m,k}$ uniformly without replacement and computes
\begin{equation}
 g_m^k=-\frac{\alpha_m}{b_{m,k}}\sum_{i\in\mathcal B_m^k}
x_{mi}\psi_{\tau,\eps}(y_{mi}-x_{mi}^{\mathsf{T}}\theta_m^k).
\label{eq:stochgrad}
\end{equation}
Let $\mathcal F_k$ be the sigma-field generated by the history before sampling at round $k$. Then
\[
\mathbb{E}[g_m^k\mid\mathcal F_k]=\alpha_m\nabla f_{m,\eps}(\theta_m^k).
\]
\begin{assumption}[Stochastic oracle]\label{ass:oracle}
There are finite $\nu_m$ such that
\[
\mathbb{E}\!\left[\|g_m^k-\alpha_m\nabla f_{m,\eps}(\theta_m^k)\|^2\mid\mathcal F_k\right]
\le \frac{\alpha_m^2\nu_m^2}{b_{m,k}}.
\]
Sampling is conditionally independent across workers.
\end{assumption}
Consequently, the stacked variance is bounded by
\begin{equation}
\sigma_k^2:=\sum_{m=1}^M\frac{\alpha_m^2\nu_m^2}{b_{m,k}}.
\label{eq:variance}
\end{equation}

\subsection{Updates}
Let $u_m$ be a scaled dual variable, $\varrho>0$ the ADMM penalty, and $\eta_{k+1}>0$ a stochastic proximal step. Linearizing only the smooth data-fitting term gives
\begin{align}
\theta_m^{k+1}=\arg\min_\theta\Bigl\{
\langle g_m^k,\theta-\theta_m^k\rangle
+\frac\varrho2\|\theta-z^k+u_m^k\|^2\notag+\frac1{2\eta_{k+1}}\|\theta-\theta_m^k\|^2\Bigr\},
\label{eq:localopt}
\end{align}
with explicit solution
\begin{equation}
\theta_m^{k+1}=\frac{\eta_{k+1}^{-1}\theta_m^k+
\varrho(z^k-u_m^k)-g_m^k}{\eta_{k+1}^{-1}+\varrho}.
\label{eq:localclosed}
\end{equation}
The coordinator computes
\begin{align}
\bar v^{k+1}=\frac1M\sum_{m=1}^M(\theta_m^{k+1}+u_m^k),\quad z^{k+1}=\prox_{P/(M\varrho)}(\bar v^{k+1}).
\label{eq:coordinator}
\end{align}
followed by
\begin{equation}
u_m^{k+1}=u_m^k+\theta_m^{k+1}-z^{k+1}.
\label{eq:dual}
\end{equation}
For $P(z)=\lambda\|z\|_1$, the coordinator uses componentwise soft thresholding at $\lambda/(M\varrho)$. Elastic-net and group-lasso updates are similarly explicit.

\begin{algorithm}[t]
\caption{Distributed stochastic smoothing ADMM (DSS-ADMM)}
\label{alg:dss}
\begin{algorithmic}[1]
\Require Local data, $\tau$, $\eps$, $P$, $\varrho$, $\{\eta_k\}$, $\{b_{m,k}\}$
\State Initialize $\theta_m^0=z^0=u_m^0=0$
\For{$k=0,1,\ldots,K-1$}
  \ForAll{workers $m$ in parallel}
    \State sample $\mathcal B_m^k$ and compute $g_m^k$ by \eqref{eq:stochgrad}
    \State compute $\theta_m^{k+1}$ by \eqref{eq:localclosed} and upload $\theta_m^{k+1}+u_m^k$
  \EndFor
  \State coordinator computes $z^{k+1}$ by \eqref{eq:coordinator} and broadcasts it
  \ForAll{workers $m$ in parallel}
    \State update $u_m^{k+1}$ by \eqref{eq:dual}
  \EndFor
\EndFor
\end{algorithmic}
\end{algorithm}

\begin{center}
\begin{minipage}{0.96\textwidth}
\centering
\includegraphics[width=.82\textwidth]{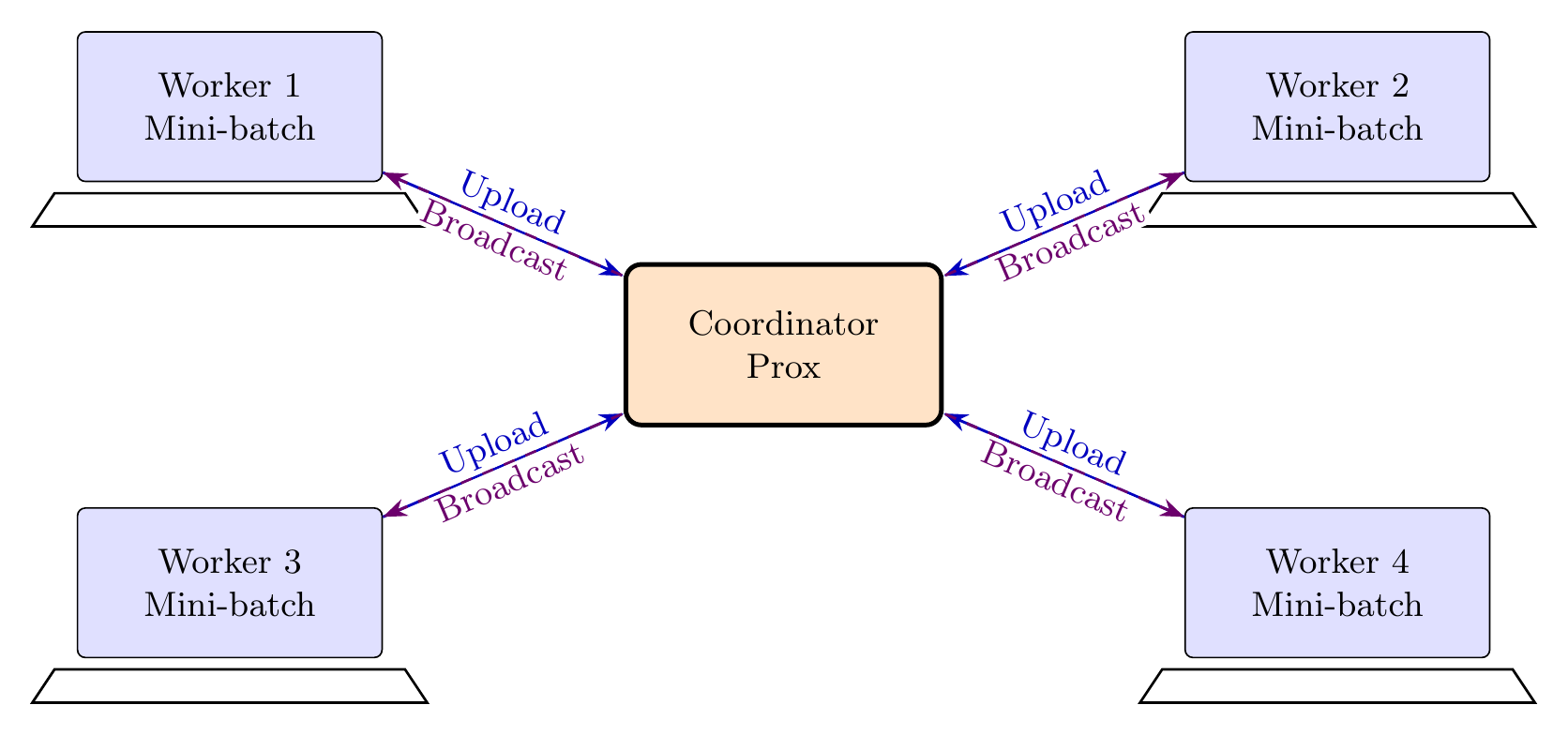}
\captionof{figure}{One communication round of DSS-ADMM.}
\label{fig:architecture}
\end{minipage}
\end{center}

Figure~\ref{fig:architecture} summarizes the complete distributed update performed at communication round $k$. Each worker retains its raw observations locally, samples the mini-batch $\mathcal B_m^k$, evaluates the derivative of the smoothed quantile loss to construct the stochastic gradient $g_m^k$, and computes the explicit local primal vector $\theta_m^{k+1}$ from \eqref{eq:localclosed}. The solid upload arrows represent transmission of the shifted vector $\theta_m^{k+1}+u_m^k$ to the coordinator; transmitting this sum avoids a separate dual-variable message. After receiving the uploads, the coordinator forms $\bar v^{k+1}=M^{-1}\sum_{m=1}^M(\theta_m^{k+1}+u_m^k)$ and applies the regularizer-dependent proximal map in \eqref{eq:coordinator} to obtain the new consensus vector $z^{k+1}$. The dashed broadcast arrows represent transmission of $z^{k+1}$ back to every worker. Each worker then completes the scaled-dual update $u_m^{k+1}=u_m^k+\theta_m^{k+1}-z^{k+1}$ and proceeds to the next round. Thus, one DSS-ADMM cycle consists of local mini-batch sampling, local stochastic-gradient evaluation, local primal updating, coordinator-side proximal aggregation, consensus broadcasting, and worker-side dual updating.

\subsection{Computational and communication cost}
The principal local cost is the matrix-vector work required to evaluate the mini-batch gradient in \eqref{eq:stochgrad}. To state the resource requirements precisely, let $C_P(p)$ denote the arithmetic cost of evaluating the coordinator proximal map $\prox_{P/(M\varrho)}$ on a $p$-vector. For the lasso and elastic net, $C_P(p)=O(p)$; group-separable penalties also have linear cost when the groups form a fixed partition.

\begin{proposition}[Per-round and cumulative resource cost]\label{prop:resource-cost}
Assume dense $p$-dimensional covariates, full participation by all $M$ workers, and point-to-point accounting in which a broadcast to $M$ workers counts as $M$ transmitted messages. At round $k$:
\begin{enumerate}[label=\textup{(\roman*)}]
\item worker $m$ uses $O(b_{m,k}p+p)$ arithmetic operations and $O(p)$ working memory beyond its local data and the sampled mini-batch;
\item the coordinator uses $O(Mp+C_P(p))$ arithmetic operations and $O(p)$ working memory if the uploaded vectors are accumulated as they arrive;
\item the network transmits $2Mp$ real scalars: each worker uploads one $p$-vector and receives one $p$-vector.
\end{enumerate}
Consequently, after $K$ rounds the total arithmetic work across all devices is
\begin{equation}
O\!\left(
p\sum_{k=0}^{K-1}\sum_{m=1}^M b_{m,k}
+K\{Mp+C_P(p)\}
\right),
\label{eq:cumulative-computation}
\end{equation}
and the cumulative communication volume is $2KMp$ real scalars under this communication model. Ignoring synchronization and network latency, the parallel arithmetic time per round is
\begin{equation}
O\!\left(
p\max_{1\le m\le M}b_{m,k}+Mp+C_P(p)
\right).
\label{eq:parallel-cost}
\end{equation}
\end{proposition}

\begin{proof}
For each sampled observation, evaluating the residual, the scalar derivative $\psi_{\tau,\eps}$, and its product with $x_{mi}$ requires $O(p)$ arithmetic. Summing over $b_{m,k}$ observations gives $O(b_{m,k}p)$ operations, while the explicit primal and dual updates require only $O(p)$ additional work. The worker stores $\theta_m^k$, $u_m^k$, and the received consensus vector $z^k$, which accounts for $O(p)$ working memory beyond its data and mini-batch. The coordinator adds $M$ uploaded $p$-vectors and applies one proximal map, giving $O(Mp+C_P(p))$ work. A running sum requires only one $p$-vector of working memory. Finally, the upload of $\theta_m^{k+1}+u_m^k$ and the broadcast of $z^{k+1}$ contribute $p+p$ transmitted scalars per worker. Summing these quantities over workers and rounds proves \eqref{eq:cumulative-computation}; taking the slowest parallel worker and then adding the sequential coordinator stage gives \eqref{eq:parallel-cost}.
\end{proof}

Proposition~\ref{prop:resource-cost} separates computation from communication. Replacing a full local gradient by a mini-batch reduces worker $m$'s leading per-round computation from $O(n_mp)$ to $O(b_{m,k}p)$, but it does not reduce the $2Mp$ scalar communication volume of a synchronous round under point-to-point accounting. Thus smaller batches save local computation at a fixed number of rounds, whereas reducing communication requires fewer rounds, partial participation, compression, multicast-aware accounting, or a different communication protocol. The convergence theory below assumes full participation and uncompressed messages, so these extensions are outside its scope.

\section{Convex convergence theory}\label{sec:convex-theory}
Stack $\beta=(\theta_1^{\mathsf{T}},\ldots,\theta_M^{\mathsf{T}})^{\mathsf{T}}$ and write $A_1=I_{Mp}, A_2=-(\mathbf1_M\otimes I_p), b=0$.
The smoothed consensus problem is a two-block linearly constrained problem with
\[
f_{1,\eps}(\beta)=\sum_{m=1}^M\alpha_m f_{m,\eps}(\theta_m),
\qquad f_2(z)=P(z),
\]
and
\begin{equation}
L_\eps=\max_m\alpha_mL_{m,\eps}
\le\max_m\frac{\alpha_m\|X_m\|_2^2}{2\eps n_m}.
\label{eq:Lstack}
\end{equation}
Thus Algorithm~\ref{alg:dss} is stochastic ADMM applied to a stacked two-block problem.

\begin{assumption}[Convex setting]\label{ass:convex}
$P$ is proper, closed, and convex; the smoothed consensus problem has a primal-dual solution; the iterates are initialized at finite values; and Assumption~\ref{ass:oracle} holds with fixed batch sizes, so $\sigma_k\le\sigma_b<\infty$.
\end{assumption}
Let $(\beta^\star,z^\star)$ be a primal solution,
$D_\beta=\|\beta^0-\beta^\star\|$, and
$D_z=\|A_2z^0-A_2z^\star\|$. For $D_\beta>0$, choose
\begin{equation}
\eta_{k+1}=\frac1{2L_\eps+\sigma_b\sqrt{k+1}/D_\beta}.
\label{eq:stepsize}
\end{equation}
Define weighted averages
\begin{align*}
\bar\beta^K=\frac{\sum_{k=0}^{K-1}\eta_{k+1}\beta^{k+1}}
{\sum_{k=0}^{K-1}\eta_{k+1}}, \quad \bar z^K=\frac{\sum_{k=0}^{K-1}\eta_{k+1}z^{k+1}}
{\sum_{k=0}^{K-1}\eta_{k+1}}.
\end{align*}

\begin{theorem}[Expected objective-feasibility convergence]\label{thm:convex-rate}
Under Assumption~\ref{ass:convex} and \eqref{eq:stepsize}, for every $q>0$ and all sufficiently large $K$, there are finite constants $C_0,C_1$ depending only on $D_\beta,D_z,L_\eps,\sigma_b,\varrho,$ and $q$ such that
\begin{align}
\mathbb{E}\bigl[&f_{1,\eps}(\bar\beta^K)+P(\bar z^K)-F_\eps^\star +q\|A_1\bar\beta^K+A_2\bar z^K\|\bigr]
\le \frac{C_0\log K+C_1}{\sqrt K}.
\label{eq:mainrate}
\end{align}
If $\lambda^\star$ is an optimal multiplier and $q>\|\lambda^\star\|$, then
\[
\mathbb{E}\|A_1\bar\beta^K+A_2\bar z^K\|=O(\log K/\sqrt K).
\]
\end{theorem}
The proof is given in Appendix~\ref{app:convex-proofs}. Conditional unbiasedness eliminates the martingale cross term, while the bounded-variance assumption controls the remaining stochastic error. The resulting rate is consistent with stochastic ADMM analyses based on noisy gradient estimates \citep{ouyang2013}.

Each smoothed local loss is Lipschitz with
\begin{align*}
G_m=\max(\tau,1-\tau)\frac1{n_m}
\sum_{i=1}^{n_m}\|x_{mi}\|, \quad G_c=\Bigl(\sum_m\alpha_m^2G_m^2\Bigr)^{1/2}.
\end{align*}
Hence $f_{1,\eps}(\mathbf1_M\otimes z)\le f_{1,\eps}(\beta)+G_c\|\beta-\mathbf1_M\otimes z\|$.

\begin{proposition}[Shared model and original check loss]\label{prop:shared-model}
Taking $q\ge G_c$ in Theorem~\ref{thm:convex-rate} gives
\[
\mathbb{E}[F_\eps(\bar z^K)-F_\eps^\star]
\le \frac{C_0\log K+C_1}{\sqrt K},
\]
and for the original nonsmooth objective \eqref{eq:target},
\begin{equation}
\mathbb{E}[F(\bar z^K)-F^\star]
\le \frac{C_0\log K+C_1}{\sqrt K}+\frac\eps4.
\label{eq:transfer}
\end{equation}
\end{proposition}
Equation~\eqref{eq:transfer} separates stochastic optimization error from smoothing error. A fixed small $\eps$ gives an explicit accuracy floor. Approaching the exact check-loss solution by $\eps_s\downarrow0$ requires increasingly accurate inner solves; a single-loop decreasing-$\eps$ claim needs a separate time-varying analysis because $L_\eps$ diverges as $\eps\downarrow0$.

\begin{proposition}[Strongly convex smooth block]\label{prop:strong-convex}
Suppose the assumptions of Theorem~\ref{thm:convex-rate} hold and, in addition, $f_{1,\eps}$ is $\mu_c$-strongly convex. Define the unweighted ergodic averages
\[
\widehat\beta^K=\frac1K\sum_{k=0}^{K-1}\beta^{k+1},
\qquad
\widehat z^K=\frac1K\sum_{k=0}^{K-1}z^{k+1}.
\]
With
\[
\eta_{k+1}=\{2L_\eps+(k+1)\mu_c\}^{-1},
\]
for every $q>0$ there is a finite constant $C_{\rm sc}$, independent of $K$, such that
\[
\mathbb E\!\left[f_{1,\eps}(\widehat\beta^K)+P(\widehat z^K)-F_\eps^\star
+q\|A_1\widehat\beta^K+A_2\widehat z^K\|\right]
\le \frac{C_{\rm sc}+\sigma_b^2\sum_{k=0}^{K-1}\eta_{k+1}}{K}
=O\!\left(\frac{\log K}{K}\right).
\]
If $q>\|\lambda^\star\|$, the feasibility residual has the same order. Strong convexity must belong to the stochastic smooth block; strong convexity of the coordinator block alone does not yield this telescope.
\end{proposition}

\begin{remark}[Variance-reduced recursive momentum and acceleration]
Because each $f_{m,\eps}$ is a differentiable finite sum, the stochastic variance-reduced recursive momentum estimator of \citet{long2026} could replace the plain local estimator $g_m^k$. This modification may reduce the persistent error associated with fixed mini-batches and is therefore relevant to the weakly convex extension, which requires vanishing estimator error. The $O(1/T)$ result of \citet{long2026}, however, is a stationarity bound for a global finite-sum constrained problem and does not automatically improve Theorem~\ref{thm:convex-rate} to an $O(1/K)$ distributed objective-feasibility rate. Establishing such an improvement would require a new Lyapunov function that incorporates the worker-specific estimator errors and the consensus residual. The accelerated estimator could likewise be incorporated into the local smooth block, but a corresponding distributed rate would require simultaneous control of extrapolation, recursive tracking error, and consensus.
\end{remark}

\section{Scope of the weakly convex extension}\label{sec:weakly-convex}
The convex framework covers the lasso, elastic net, group lasso, and convex constraints. To reduce shrinkage bias, consider
\[
P_{\rm nc}(z)=\sum_{j=1}^p p_{\lambda,a}(|z_j|),
\]
where $p_{\lambda,a}$ is MCP \citep{zhang2010} with $a>1$ or SCAD \citep{fanli2001} with $a>2$. These penalties are weakly convex:
\[
\kappa_{\rm MCP}=1/a,
\qquad \kappa_{\rm SCAD}=1/(a-1).
\]
\begin{proposition}[Well-posed coordinator update]\label{prop:coordinator-update}
If $M\varrho>\kappa$, then
\[
P_{\rm nc}(z)+\frac{M\varrho}{2}\|z-\bar v^{k+1}\|^2
\]
is strongly convex. Hence the MCP or SCAD coordinator update is single-valued and can be evaluated coordinatewise.
\end{proposition}
Proposition~\ref{prop:coordinator-update} establishes well-posedness of the coordinator subproblem but does not imply global optimality of the full nonconvex problem. The proof of Theorem~\ref{thm:convex-rate} uses the convex subgradient inequality and Jensen's inequality, neither of which applies directly to MCP or SCAD. Moreover, with a constant step size, fixed-mini-batch noise generally permits convergence only to a neighborhood of stationarity. The next result is deliberately stated as a conditional implication: it isolates sufficient descent and relative-error properties, rather than claiming that Algorithm~\ref{alg:dss} with a fixed mini-batch automatically satisfies them.

\begin{theorem}[Conditional KKT implication under inexact-ADMM estimates]\label{thm:conditional-stationarity}
Let $e^k=\widetilde\nabla f_{1,\eps}(\beta^k)-\nabla f_{1,\eps}(\beta^k)$ and $w^k=(\beta^k,z^k,\lambda^k)$. Suppose that $P_{\rm nc}$ is proper, lower semicontinuous, and $\kappa$-weakly convex, that $f_{1,\eps}$ has an $L_\eps$-Lipschitz gradient, and that $M\varrho>\kappa$. Assume that the generated sequence is bounded and that there exists a lower-bounded proximal-ADMM Lyapunov function $\Phi$ whose critical points coincide with the KKT points and that satisfies, almost surely,
\begin{align}
\Phi(w^{k+1})\le\Phi(w^k)-c_1\|w^{k+1}-w^k\|^2+c_2\|e^k\|^2,
\label{eq:ncdescent}\\
\dist(0,\partial\Phi(w^{k+1}))\le c_3\|w^{k+1}-w^k\|+c_3\|e^k\|.
\label{eq:ncrelative}
\end{align}
If $\sum_k\|e^k\|^2<\infty$, then $\|w^{k+1}-w^k\|\to0$, the consensus residual converges to zero, and every cluster point satisfies the limiting-subdifferential KKT system. If, in addition, $\Phi$ has the Kurdyka--{\L}ojasiewicz property and $\sum_k\|e^k\|<\infty$, the whole sequence has finite length and converges to one KKT point.
\end{theorem}
Here the KKT system is
\begin{align*}
0&=\nabla f_{1,\eps}(\beta^\infty)+A_1^{\mathsf{T}}\lambda^\infty,\\
0&\in\partial P_{\rm nc}(z^\infty)+A_2^{\mathsf{T}}\lambda^\infty,\\
0&=A_1\beta^\infty+A_2z^\infty.
\end{align*}
Theorem~\ref{thm:conditional-stationarity} is an implication under the stated descent and relative-error estimates; it does not assert that fixed mini-batches satisfy them automatically. Such conditions are standard in inexact nonconvex proximal-ADMM analyses \citep{bolte2014,hong2016,wang2019,rockafellar1998}. Full gradients give $e^k=0$. A finite-data batch schedule that eventually reaches the full local sample also makes $e^k=0$ after finitely many iterations, whereas a fixed mini-batch generally does not. Appendix~\ref{app:weakly-convex} gives the conditional argument and clarifies sufficient error regimes.

\section{Numerical experiments}\label{sec:experiments}
\subsection{Implementation and evaluation design}
The numerical studies use a single-process emulation of synchronous horizontally distributed learning: workers retain their observations, independently sample local mini-batches, and exchange only coefficient vectors. Reported wall-clock times therefore measure computation in this serial emulation; per-observation gradient evaluations and communication rounds provide implementation-independent workload measures.

The experiments were run under Linux x86\_64 on an Intel Xeon Platinum 8370C CPU (5 vCPUs), using Python 3.13.5, NumPy 2.3.5, pandas 2.2.3, scikit-learn 1.8.0, and statsmodels 0.14.6. All random seeds were fixed. Synthetic predictors were generated from a Gaussian distribution with AR(1) covariance $\Sigma_{jk}=0.3^{|j-k|}$, and the first $s$ coefficients were nonzero. Responses were generated with heteroscedastic Student-$t_3$ or contaminated-normal errors. Each synthetic data set was split into 75\% training observations and 25\% test observations before fitting. Predictor columns and the response were centered and scaled using the training-sample means and standard deviations, and the same transformations were applied to the test sample; consequently, the reported pinball losses are on the standardized response scale.

Unless stated otherwise, $\tau=0.5$, $M=8$, $\eps=0.08$, $\varrho=1$, and lasso regularization was used. The implementation uses the practical schedule $\eta_{k+1}=0.98\{2L_\eps+c\sqrt{k+1}\}^{-1}$, with $c=0.12$ in the synthetic experiments and $c=0.10$ in the two real-data studies. This has the same $k^{-1/2}$ decay as the stochastic schedule used in Theorem~\ref{thm:convex-rate}, while avoiding the unknown theoretical constant $D_\beta/\sigma_b$. No tolerance-based stopping rule is used in the reported experiments; each run uses the fixed communication horizon stated in the corresponding subsection. The reported training objective is the original nonsmoothed penalized objective $F(z)$, the test loss is the unpenalized mean check loss, and the consensus diagnostic is $\{M^{-1}\sum_{m=1}^M\|\theta_m-z\|^2\}^{1/2}$. A coefficient is counted as selected when $|z_j|>10^{-6}$; selection F1 scores use the $10^{-5}$ support threshold implemented in the released code. We compare full-batch and mini-batch DSS-ADMM, distributed proximal SGD, and the centralized linear-programming implementation in \texttt{QuantileRegressor} \citep{pedregosa2011}.

\subsection{Stochastic computation and consensus}
Figure~\ref{fig:stoch} shows one representative run with $n=2600$, $p=120$, $\lambda=0.018$, and 600 communication rounds. Over the plotted horizon, the objective values of the mini-batch variants stabilize after substantially fewer per-observation gradient evaluations than the full-batch variant uses. Smaller batches reduce the cost per round but yield a larger consensus residual at the fixed termination horizon. Distributed proximal SGD is included in the objective panel as a first-order baseline; because it maintains a single global iterate rather than worker-specific primal copies, the DSS-ADMM consensus diagnostic is not directly comparable for that method.

\begin{center}
\begin{minipage}{0.96\textwidth}
\centering
\includegraphics[width=.94\textwidth]{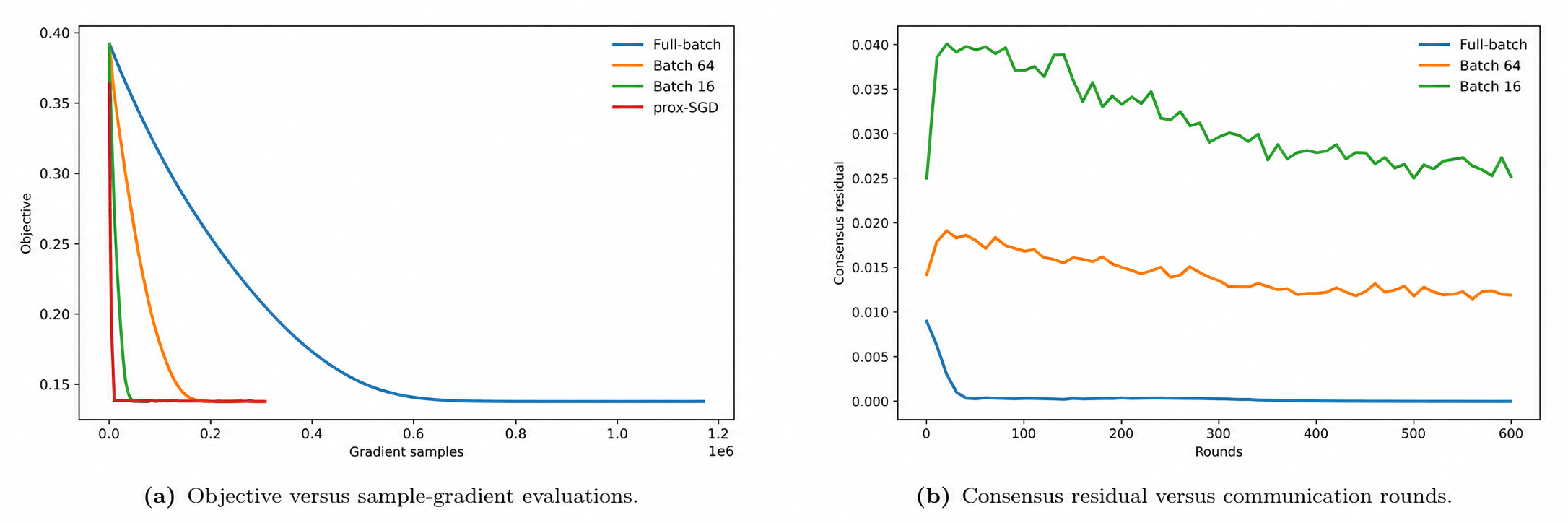}
\captionof{figure}{A representative heavy-tailed simulation: objective versus sampled-gradient evaluations and consensus residual versus communication rounds.}
\label{fig:stoch}
\end{minipage}
\end{center}

\begin{table*}[t]
\caption{Contaminated-error simulations: mean (standard deviation) over five replications.}
\label{tab:main}
\centering
\small
\resizebox{\textwidth}{!}{%
\begin{tabular}{lrrrrrr}
\toprule
Method & Train obj. & Test loss & Consensus & Selected & Samples & Time (s)\\
\midrule
Full-batch DSS-ADMM & .1637 (.0118) & .1148 (.0120) & .000002 & 16.20 & 900000 & .274\\
DSS-ADMM, batch 64 & .1637 (.0118) & .1148 (.0119) & .011313 & 16.40 & 307200 & .202\\
DSS-ADMM, batch 16 & .1641 (.0118) & .1150 (.0118) & .026160 & 17.00 & 76800 & .171\\
Centralized LP & .1636 (.0118) & .1147 (.0119) & 0 & 15.20 & -- & --\\
\bottomrule
\end{tabular}
}
\end{table*}

Table~\ref{tab:main} summarizes five contaminated-error replications with $n=2000$, $p=100$, $s=10$, $\lambda=0.02$, and 600 communication rounds. The batch-64 variant used roughly one third as many per-observation gradient evaluations as the full-batch variant while achieving a similar test pinball loss. The batch-16 variant used less than one tenth as many evaluations but produced a larger consensus residual and a denser estimate at the fixed termination horizon. These results demonstrate a computational trade-off; they do not imply that smaller batches improve variable selection.

\subsection{Partition heterogeneity and scaling}
For the non-IID partition, observation $i$ is assigned the score $s_i=x_{i1}+0.35x_{i2}+0.05\xi_i$, where $\xi_i\sim N(0,1)$ is independent jitter. Observations are sorted by $s_i$ and contiguous score bands are assigned to workers, producing a controlled covariate shift. Table~\ref{tab:iid} reports similar objective values and test losses for the IID and non-IID partitions. Thus, in this experiment, the moderate covariate shift had little effect under full worker participation.

\begin{table}[ht]
\caption{IID and non-IID partitions: mean (standard deviation) over four replications.}
\label{tab:iid}
\centering
\footnotesize
\setlength{\tabcolsep}{3.5pt}
\begin{tabular}{lrrr}
\toprule
Partition & Train & Test & Consensus\\
\midrule
IID & .1131 (.0092) & .06518 (.00306) & .01051\\
Non-IID & .1132 (.0092) & .06522 (.00302) & .00970\\
\bottomrule
\end{tabular}
\end{table}

Figure~\ref{fig:scale} reports emulated computation times for $p=80$ and 350 rounds. At $n=7200$, batch 32 accessed 89,600 observations, batch 128 accessed 358,400, and the full-batch method accessed 1,890,000. The corresponding objectives were 0.12999, 0.12990, and 0.12988. These timings do not include an actual communication network.

\begin{center}
\begin{minipage}{0.96\textwidth}
\centering
\includegraphics[width=0.8\columnwidth]{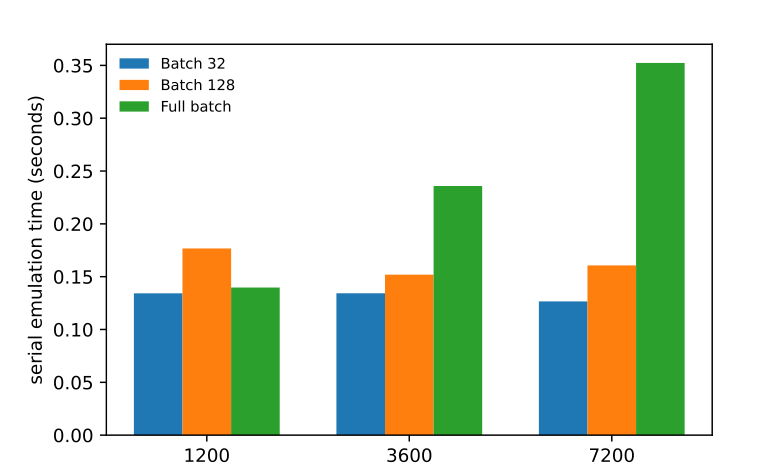}
\captionof{figure}{Serial emulation time over 350 rounds.}
\label{fig:scale}
\end{minipage}
\end{center}

\subsection{Worker count and smoothing sensitivity}
We vary the number of workers while keeping the total mini-batch budget at 128 observations per round. Table~\ref{tab:workers} shows stable test loss for $M=2,4,8,16$. Communicated vectors per round grow with $M$, while the finite-horizon consensus residual remains of the same order. The smoothing parameter controls both approximation and conditioning. In Figure~\ref{fig:sens}, very small smoothing gives a larger finite-horizon loss because the gradient Lipschitz constant increases; values between 0.05 and 0.12 perform similarly.

\begin{table}[ht]
\caption{Worker sensitivity: mean over four replications.}
\label{tab:workers}
\centering
\small
\begin{tabular}{rrrr}
\toprule
$M$ & Batch/worker & Test loss & Consensus\\
\midrule
2 & 64 & .06494 & .02550\\
4 & 32 & .06462 & .02598\\
8 & 16 & .06492 & .02514\\
16 & 8 & .06495 & .02295\\
\bottomrule
\end{tabular}
\end{table}

\begin{center}
\begin{minipage}{0.96\textwidth}
\centering
\includegraphics[width=.94\textwidth]{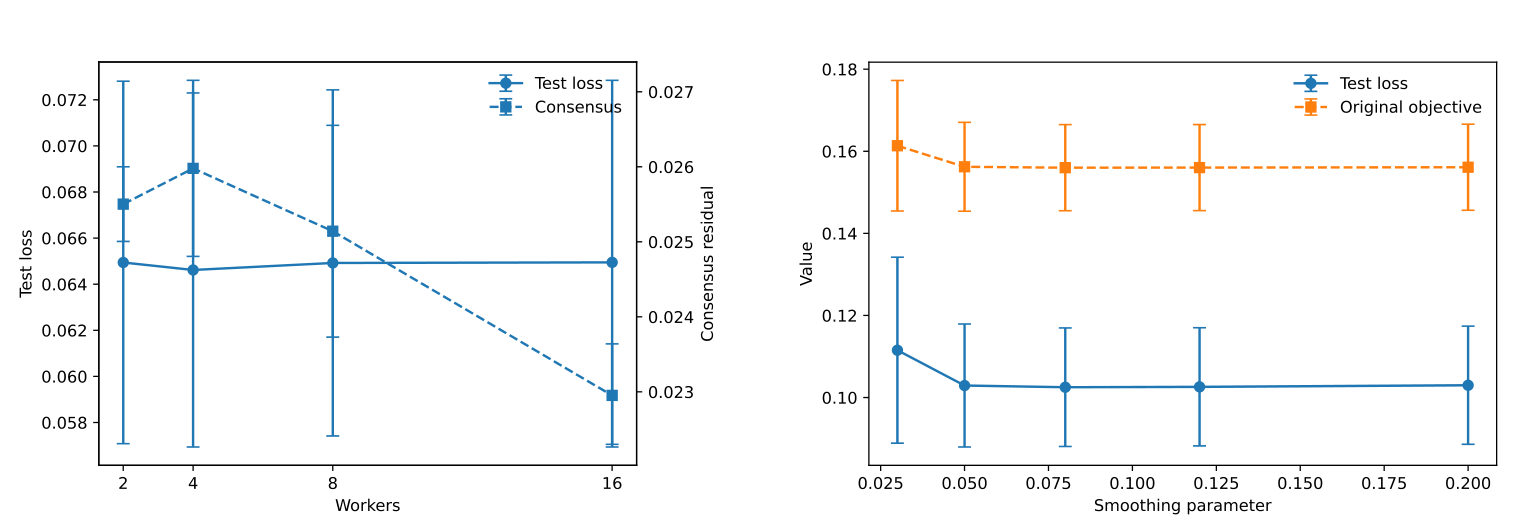}
\captionof{figure}{Sensitivity to the number of workers and the smoothing parameter. Error bars show one standard deviation.}
\label{fig:sens}
\end{minipage}
\end{center}

\subsection{Regularization and mini-batch interaction}
Figure~\ref{fig:heat} compares five lasso regularization levels and three batch regimes over five replications. Smaller $\lambda$ values yield lower prediction loss but select more noise variables. Values near $\lambda=0.04$ produce nearly perfect support recovery for the batch-64 and full-batch methods, with a modest increase in test loss. In this experiment, mini-batch size affects finite-horizon consensus and selection accuracy more strongly than prediction loss.

\begin{center}
\begin{minipage}{0.96\textwidth}
\centering
\includegraphics[width=.88\textwidth]{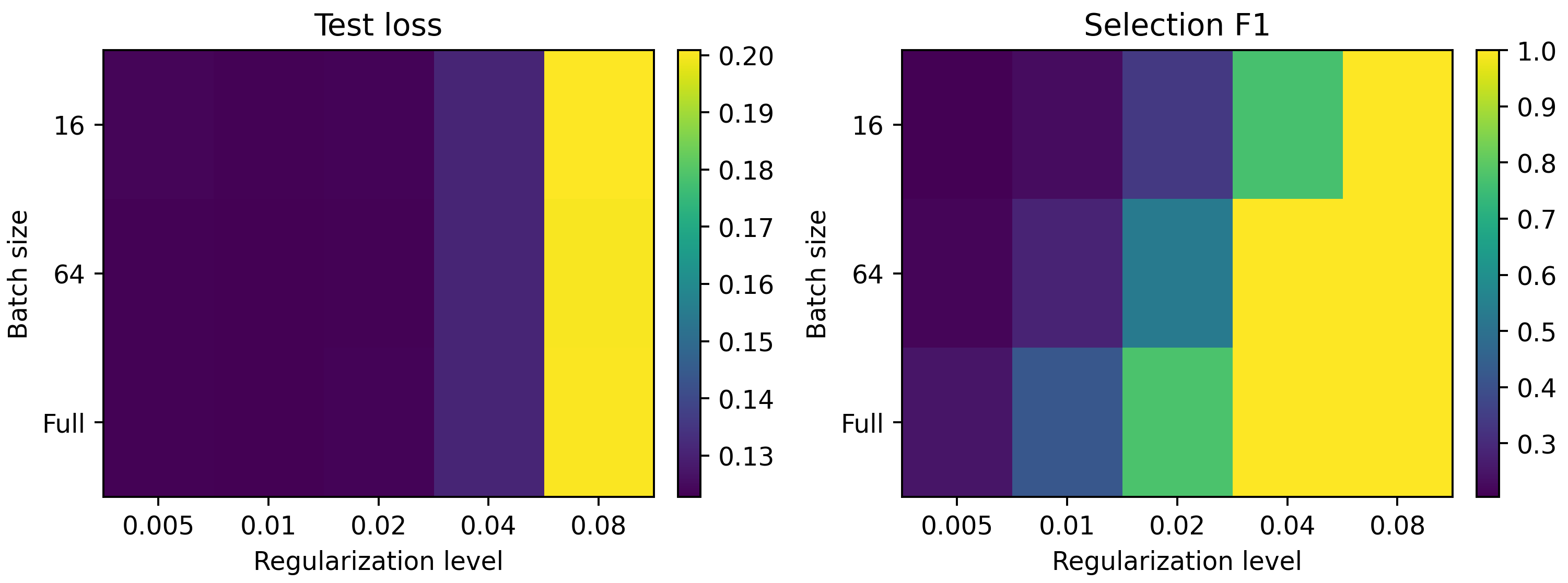}
\captionof{figure}{Joint sensitivity to the lasso regularization level and mini-batch size. Entries are means over five replications.}
\label{fig:heat}
\end{minipage}
\end{center}

\subsection{Error robustness and feature dimension}
We consider Gaussian, Student-$t_3$, Laplace, and contaminated-normal errors at three quantile levels. Table~\ref{tab:errors} reports the test pinball loss. Across the settings considered, the losses vary only moderately across quantile levels, whereas contaminated errors produce the largest losses. We also vary the number of features from 40 to 320 under a fixed per-observation gradient-evaluation budget. Figure~\ref{fig:robust} shows a moderate increase in serial emulation time and consensus residual as the dimension increases.

\begin{table}[t]
\caption{Test loss under four error distributions: mean over three replications.}
\label{tab:errors}
\centering
\small
\begin{tabular}{lrrr}
\toprule
Error & $\tau=.25$ & $\tau=.50$ & $\tau=.75$\\
\midrule
Gaussian & .08018 & .08093 & .08177\\
Student-$t_3$ & .08350 & .08337 & .08326\\
Laplace & .07898 & .07675 & .07459\\
Contaminated & .13121 & .13194 & .13266\\
\bottomrule
\end{tabular}
\end{table}

\begin{center}
\begin{minipage}{0.96\textwidth}
\centering
\includegraphics[width=.94\textwidth]{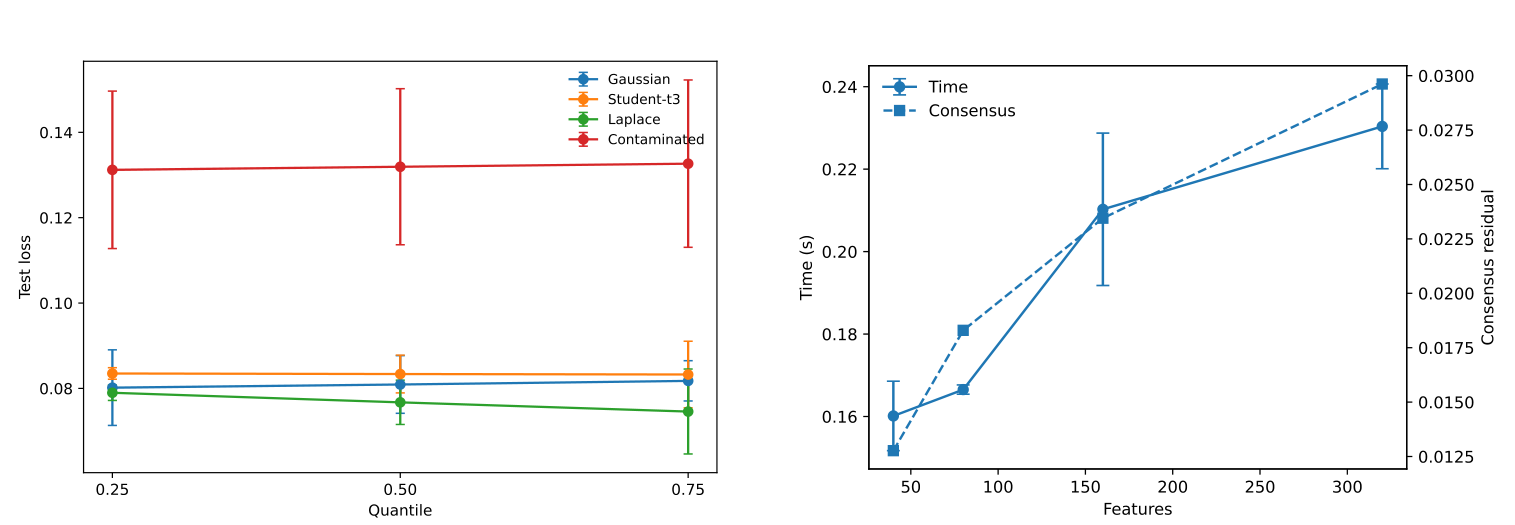}
\captionof{figure}{Robustness across error distributions and dimensional scaling.}
\label{fig:robust}
\end{minipage}
\end{center}

\subsection{Unequal site sizes and predictor correlation}
The formulation accommodates unequal site sizes through $\alpha_m=n_m/N$. We compare balanced local samples with moderate and severe imbalance while holding the total sample size, number of workers, batch size, and number of rounds fixed. The largest-to-smallest local-size ratios are 1, 2, and 8. Across five replications, the mean test losses are 0.06689, 0.06689, and 0.06720, and the mean consensus residuals are 0.01793, 0.01697, and 0.01765. In these experiments, sample-size weighting is associated with nearly unchanged predictive accuracy despite substantial imbalance.

We also vary the AR(1) predictor correlation from 0 to 0.9 under contaminated errors. Test loss remains near 0.12 for correlations up to 0.3, then rises to 0.1398 at correlation 0.6 and 0.1952 at correlation 0.9. Selection F1 decreases from approximately 0.96--0.97 to 0.78, reflecting the statistical difficulty of support recovery under severe collinearity rather than a failure of distributed consensus. Figure~\ref{fig:extra} jointly summarizes the site-imbalance and predictor-correlation sensitivity results.

\begin{center}
\begin{minipage}{0.96\textwidth}
\centering
\includegraphics[width=.94\textwidth]{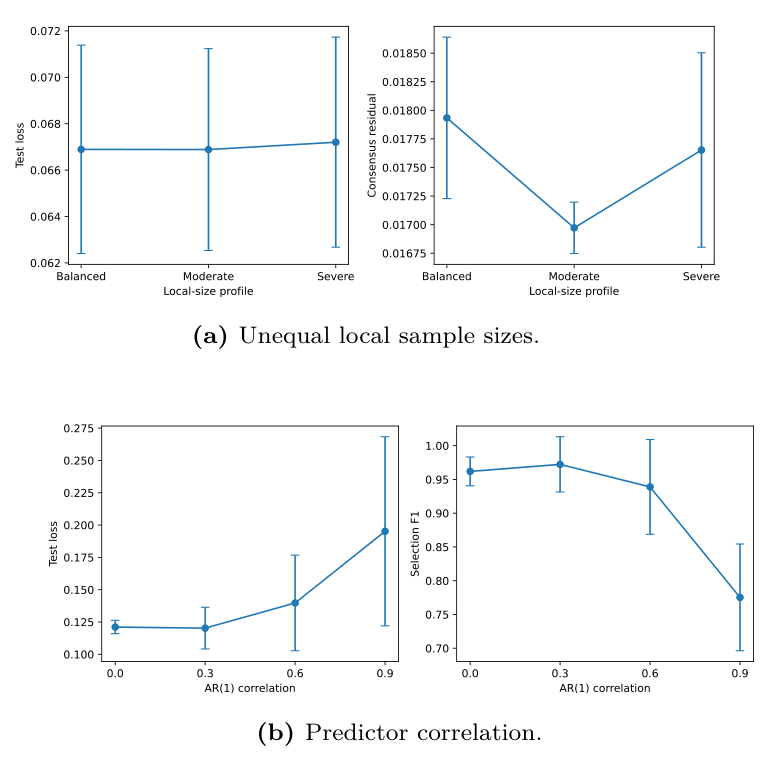}
\captionof{figure}{Additional distributed and statistical sensitivity experiments. Error bars show one standard deviation over five replications.}
\label{fig:extra}
\end{minipage}
\end{center}

\subsection{Weakly convex extension}
We evaluate lasso, MCP, and SCAD using full local gradients, which constitute a deterministic special case satisfying the vanishing-noise condition in Theorem~\ref{thm:conditional-stationarity}. In four contaminated-error replications with $n=1200$ and $p=80$, the mean test pinball losses were 0.12481, 0.12383, and 0.12373, respectively. Because the penalty functions differ, their penalized objective values are not directly comparable; predictive loss and selection measures provide more meaningful cross-penalty comparisons.

To examine the stochastic-error requirement, we compare a fixed batch of 16, the increasing schedule $b_k=\min\{n_m,\lceil8\sqrt{k+1}\rceil\}$, and full local gradients. The residual in Figure~\ref{fig:kkt} combines the full-gradient first-block KKT residual, consensus residual, and coordinator proximal residual. Fixed batches remain in a stationarity neighborhood, whereas increasing batches nearly match the full-gradient residual while using fewer per-observation gradient evaluations. Across five replications, the mean residuals for MCP versus SCAD were 0.03592 versus 0.03558 with fixed batches, 0.002525 versus 0.002528 with increasing batches, and 0.002523 versus 0.002527 with full gradients. Thus, MCP has a slight residual advantage under the increasing-batch and full-gradient regimes, while SCAD is marginally better under the fixed-batch regime. These differences are small relative to the reported standard deviations, so the experiment supports comparable stationarity behavior rather than a general ranking of the two penalties.

\begin{center}
\begin{minipage}{0.96\textwidth}
\centering
\includegraphics[width=\columnwidth]{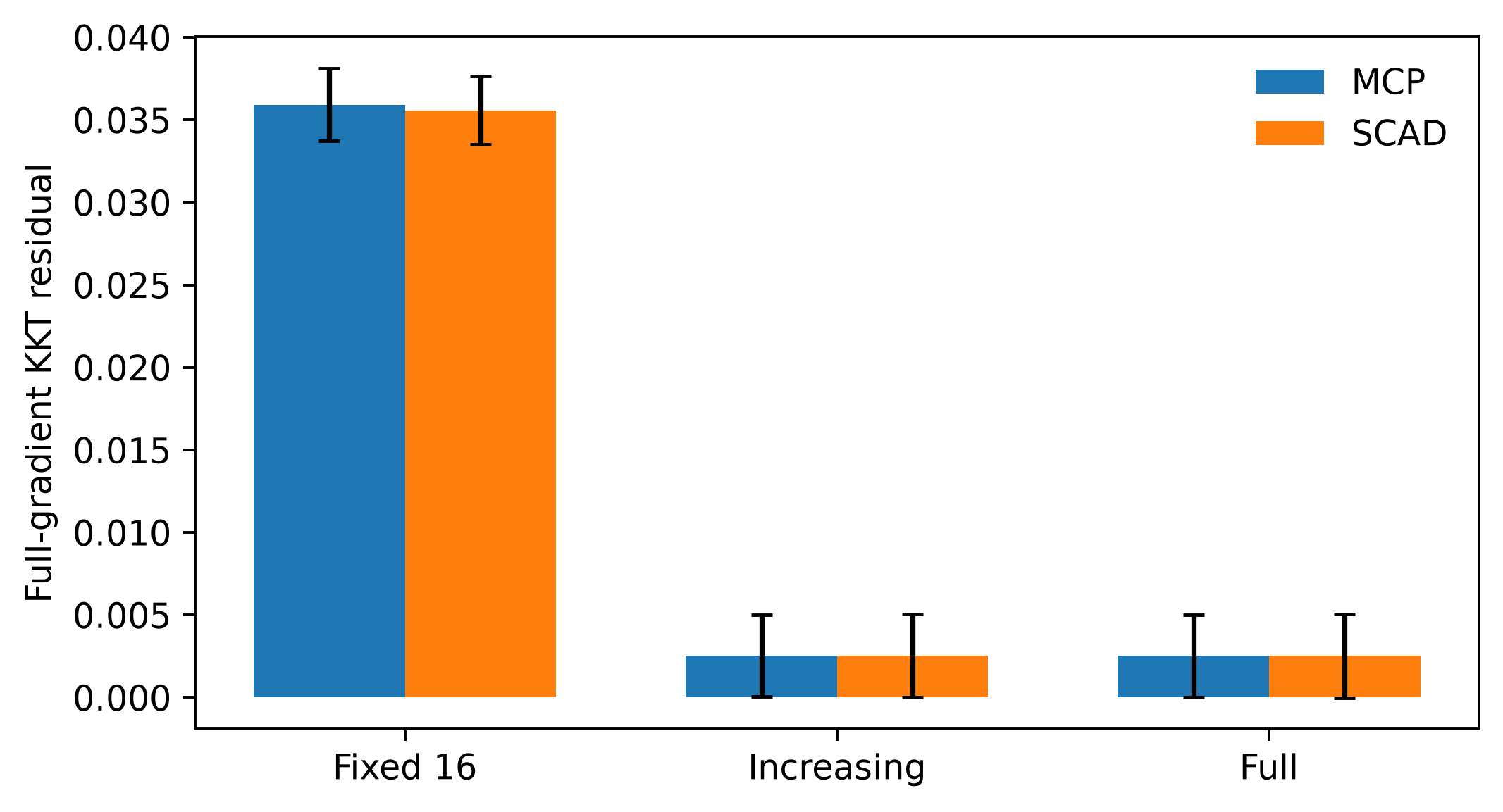}
\captionof{figure}{Full-gradient KKT residuals for MCP and SCAD under fixed, increasing, and full-batch schedules. MCP is marginally lower for the increasing and full schedules, whereas SCAD is marginally lower for the fixed schedule. Error bars show one standard deviation over five replications.}
\label{fig:kkt}
\end{minipage}
\end{center}

\subsection{Diabetes data}
The diabetes data contain 442 observations and ten baseline variables. We use eight random 75/25 splits, four non-IID workers, and quantile levels 0.25, 0.50, and 0.75. Within each split, predictors and the response are standardized using training-sample means and standard deviations only, so Table~\ref{tab:diabetes} reports pinball loss on the standardized response scale. We use $\lambda=0.012$, $\eps=0.06$, batch size 16 per worker, $\varrho=1$, and 650 communication rounds. DSS-ADMM yields test losses close to those of the centralized LP benchmark at each quantile, with a consensus residual of order $10^{-2}$ at the fixed horizon.

\begin{table}[t]
\caption{Diabetes test pinball loss: mean (standard deviation) over eight splits.}
\label{tab:diabetes}
\centering
\small
\begin{tabular}{rrr}
\toprule
$\tau$ & DSS-ADMM & Centralized LP\\
\midrule
0.25 & .2805 (.0166) & .2824 (.0164)\\
0.50 & .2834 (.0126) & .2849 (.0118)\\
0.75 & .2862 (.0216) & .2873 (.0203)\\
\bottomrule
\end{tabular}
\end{table}

\subsection{Engel data}
As an additional real-data study, we analyze the Engel expenditure data from \texttt{statsmodels}. The response is household food expenditure and the predictor block consists of six polynomial features of log income. We use eight random 70/30 splits, four non-IID workers, and quantile levels 0.25, 0.50, and 0.75. Within each split, the six predictor columns and the response are standardized using training-sample means and standard deviations only; Table~\ref{tab:engel} therefore reports pinball loss on the standardized response scale. We use $\lambda=0.010$, $\eps=0.06$, batch size 16 per worker, $\varrho=1$, and 650 communication rounds.

\begin{table}[t]
\caption{Engel data test pinball loss: mean (standard deviation) over eight splits.}
\label{tab:engel}
\centering
\small
\begin{tabular}{rrr}
\toprule
$\tau$ & DSS-ADMM & Centralized LP\\
\midrule
0.25 & .1320 (.0183) & .1372 (.0275)\\
0.50 & .1313 (.0131) & .1344 (.0188)\\
0.75 & .1306 (.0192) & .1315 (.0202)\\
\bottomrule
\end{tabular}
\end{table}

\section{Discussion and conclusion}\label{sec:conclusion}
DSS-ADMM combines smoothing of the check loss, stochastic approximation, and distributed consensus in a single inversion-free procedure. Stacking all worker models into one first block makes the consensus formulation a standard two-block stochastic ADMM problem, while the coordinator proximal map keeps convex regularization modular. The convex analysis controls the objective gap and consensus violation jointly, and the explicit $\eps/4$ approximation bound connects the smoothed problem to the original nonsmooth objective.

The numerical results suggest that mini-batches can substantially reduce per-observation gradient evaluations while maintaining predictive performance across the heterogeneous partitions and error distributions considered here, as well as in the two real-data examples. Fixed smoothing creates a nonzero approximation floor, and a decreasing smoothing schedule therefore requires a time-varying analysis. Fixed mini-batches do not in general yield exact stationarity for MCP or SCAD; increasing batch sizes, full gradients, or variance reduction are needed to satisfy the vanishing-error conditions in the nonconvex result.

Partial participation, delayed communication, compression, secure aggregation, differential privacy, and vertical feature partitioning would introduce additional bias, variance, or splitting constraints. From a statistical perspective, DSS-ADMM estimates one common conditional quantile model across sites. Personalized or hierarchical quantile models may be preferable when site-specific effects are substantial. Subject to this interpretation, the framework provides a scalable approach to convexly penalized distributed quantile regression and a conditional extension to weakly convex penalties: MCP and SCAD are computationally admissible when coordinator curvature dominates weak concavity, whereas the convergence guarantee concerns KKT stationarity under vanishing stochastic error rather than global optimality.

\section*{Data availability}
Synthetic data are generated by the released code. The diabetes data are distributed with scikit-learn, and the Engel data are distributed with statsmodels. No confidential or restricted data were used.

\section*{Code availability}
The source code, experiment drivers, and figure-generation scripts accompanying this article are publicly available at \url{https://github.com/xfwu1016/DSS-ADMM}.

\appendix
\section{Proofs for the convex results}\label{app:convex-proofs}
\subsection{Stacked formulation and smoothing bounds}
For fixed $\eps>0$, define
\begin{align*}
f_{m,\eps}(\theta)=\frac1{n_m}\sum_{i=1}^{n_m}
\rho_{\tau,\eps}(y_{mi}-x_{mi}^{\mathsf{T}}\theta),\quad f_{1,\eps}(\beta)=\sum_{m=1}^M\alpha_mf_{m,\eps}(\theta_m).
\end{align*}
With $A_1=I_{Mp}$ and $A_2=-(\mathbf1_M\otimes I_p)$, the smoothed distributed problem is
\begin{equation}
\min_{\beta,z}f_{1,\eps}(\beta)+P(z)
\quad\text{s.t.}\quad A_1\beta+A_2z=0.
\label{eq:app_problem}
\end{equation}
The $M$ local models constitute one stacked first block; hence \eqref{eq:app_problem} is a two-block problem.

Let $\mathcal F_k$ be the history before the mini-batches at iteration $k$ are drawn and define
\[
\widetilde\nabla f_{1,\eps}(\beta^k)=((g_1^k)^{\mathsf{T}},\ldots,(g_M^k)^{\mathsf{T}})^{\mathsf{T}}.
\]
Then
\begin{align}
\mathbb{E}_k[\widetilde\nabla f_{1,\eps}(\beta^k)]=\nabla f_{1,\eps}(\beta^k),\quad \mathbb{E}_k\|\widetilde\nabla f_{1,\eps}(\beta^k)-\nabla f_{1,\eps}(\beta^k)\|^2\le\sigma_b^2.
\label{eq:app_oracle}
\end{align}

\begin{lemma}[Smoothing bounds]\label{lem:app-smoothing}
For every $u\in\mathbb{R}$, $0\le\rho_{\tau,\eps}(u)-\rho_\tau(u)\le\eps/4$. Moreover, $\psi_{\tau,\eps}$ is $1/(2\eps)$-Lipschitz and
\[
\|\nabla f_{m,\eps}(\theta)-\nabla f_{m,\eps}(\theta')\|
\le\frac{\|X_m\|_2^2}{2\eps n_m}\|\theta-\theta'\|.
\]
\end{lemma}
\begin{proof}
For $|u|<\eps$,
\[
h_\eps(u)-|u|=\frac{(\eps-|u|)^2}{2\eps},
\]
and the difference is zero outside the smoothing interval. Its maximum is $\eps/2$, and the factor $1/2$ in the smoothed check loss gives $\eps/4$. The derivative of $h_\eps$ is $u/\eps$ inside and $\sign(u)$ outside; it is therefore $1/\eps$-Lipschitz. Composition with the residual map and summation yield the matrix bound.
\end{proof}

Set
\begin{align*}
\delta^k=\widetilde\nabla f_{1,\eps}(\beta^k)-\nabla f_{1,\eps}(\beta^k),\quad r^{k+1}=A_1\beta^{k+1}+A_2z^{k+1},
\end{align*}
and let $\lambda^k=\varrho u^k$ be the unscaled multiplier. For compactness, define
\begin{align*}
\Delta_\beta^k&=\|\beta^k-\beta^\star\|^2-\|\beta^{k+1}-\beta^\star\|^2,\\
\Delta_\lambda^k(\widetilde\lambda)&=\|\widetilde\lambda-\lambda^k\|^2-\|\widetilde\lambda-\lambda^{k+1}\|^2,\\
\Delta_z^k&=\|A_2z^k-A_2z^\star\|^2-\|A_2z^{k+1}-A_2z^\star\|^2.
\end{align*}

\subsection{One-step inequality}
\begin{lemma}\label{lem:one-step}
Suppose $f_{1,\eps}$ is convex and $L_\eps$-smooth, $P$ is proper closed convex, and $0<\eta_{k+1}\le1/(2L_\eps)$. For every test multiplier $\widetilde\lambda$,
\begin{align}
f_{1,\eps}(\beta^{k+1})+P(z^{k+1})-F_\eps^\star
+\langle\widetilde\lambda,r^{k+1}\rangle \le & \eta_{k+1}\|\delta^k\|^2
+\frac{\Delta_\beta^k}{2\eta_{k+1}}
+\langle-\delta^k,\beta^k-\beta^\star\rangle \\
& +\frac{\Delta_\lambda^k(\widetilde\lambda)}{2\varrho}
+\frac\varrho2\Delta_z^k.
\label{eq:onestep}
\end{align}
\end{lemma}
\begin{proof}
Let $d^k=\beta^{k+1}-\beta^k$. Smoothness at $\beta^k$ and convexity at $\beta^\star$ give
\begin{align*}
f_{1,\eps}(\beta^{k+1})-f_{1,\eps}(\beta^\star)
\le\langle\widetilde\nabla f_{1,\eps}(\beta^k),
\beta^{k+1}-\beta^\star\rangle -\langle\delta^k,\beta^k-\beta^\star\rangle-\langle\delta^k,d^k\rangle+\frac{L_\eps}{2}\|d^k\|^2.
\end{align*}
The first-block optimality condition is
\begin{align*}
0={}\widetilde\nabla f_{1,\eps}(\beta^k)+A_1^{\mathsf{T}}\lambda^k
+\varrho A_1^{\mathsf{T}}(A_1\beta^{k+1}+A_2z^k)
+\eta_{k+1}^{-1}d^k.
\end{align*}
The second-block condition and multiplier update give
\begin{align*}
0\in\partial P(z^{k+1})+A_2^{\mathsf{T}}\lambda^{k+1}, \quad \lambda^{k+1}-\lambda^k=\varrho r^{k+1}.
\end{align*}
Pair the first condition with $\beta^{k+1}-\beta^\star$, use the convex subgradient inequality for $P$, add $\langle\widetilde\lambda,r^{k+1}\rangle$, and apply
$2\langle a-b,a-c\rangle=\|a-b\|^2+\|a-c\|^2-\|b-c\|^2$.
Finally,
\[
-\langle\delta^k,d^k\rangle\le\eta_{k+1}\|\delta^k\|^2+\frac{1}{4\eta_{k+1}}\|d^k\|^2,
\]
and the remaining coefficient of $\|d^k\|^2$ is nonpositive because $\eta_{k+1}\le1/(2L_\eps)$.
\end{proof}

\subsection{Proof of Theorem~\ref{thm:convex-rate}}
Let $S_K=\sum_{k=0}^{K-1}\eta_{k+1}$. Multiply \eqref{eq:onestep} by $\eta_{k+1}$ and sum. The first-block distance telescopes exactly. For a nonnegative sequence $a_k$ and nonincreasing $\eta_k$,
\[
\sum_{k=0}^{K-1}\eta_{k+1}(a_k-a_{k+1})\le\eta_1a_0,
\]
which bounds the weighted multiplier and $A_2z$ differences. Convexity moves weighted function values to $(\bar\beta^K,\bar z^K)$.

The pathwise inequality holds for every test multiplier. After a path is realized, choose
\begin{align*}
\widetilde\lambda=
\begin{cases}
q\bar r^K/\|\bar r^K\|,&\bar r^K\ne0,\\
0,&\bar r^K=0,
\end{cases} \quad
\bar r^K=A_1\bar\beta^K+A_2\bar z^K.
\end{align*}
This choice is legitimate because the only martingale term does not contain $\widetilde\lambda$. Taking expectations and using conditional unbiasedness yields
\small{\begin{align}
\mathbb{E}\left[f_{1,\eps}(\bar\beta^K)+P(\bar z^K)-F_\eps^\star
+q\|\bar r^K\|\right] \le\frac1{S_K}\left[
\sigma_b^2\sum_{k=0}^{K-1}\eta_{k+1}^2+\frac{D_\beta^2}{2}
+\frac{\eta_1q^2}{2\varrho}
+\frac{\varrho\eta_1D_z^2}{2}\right].
\label{eq:masterbound}
\end{align}}
For $\sigma_b>0$, the step \eqref{eq:stepsize} satisfies
$\sum_{k<K}\eta_{k+1}^2=O(\log K)$ and $S_K=\Omega(\sqrt K)$ by integral comparison. This proves \eqref{eq:mainrate}. The martingale cancellation is $\mathbb{E}\langle\delta^k,\beta^k-\beta^\star\rangle
=\mathbb{E}\langle\mathbb{E}_k\delta^k,\beta^k-\beta^\star\rangle=0$.
If $\sigma_b=0$, the variance term vanishes and the same formula uses a constant step, so $S_K=\Theta(K)$ and \eqref{eq:masterbound} gives deterministic $O(1/K)$. If $q>\|\lambda^\star\|$, the saddle-point inequality yields feasibility at the same stochastic order.

\subsection{Shared model and original check loss}
Because the derivative of the smoothed check loss lies in $[\tau-1,\tau]$,
\[
f_{1,\eps}(\mathbf1_M\otimes z)
\le f_{1,\eps}(\beta)+G_c\|\beta-\mathbf1_M\otimes z\|.
\]
Hence the feasibility term in Theorem~\ref{thm:convex-rate} absorbs the loss change incurred by replacing local averages with the shared model. For any $z$, Lemma~\ref{lem:app-smoothing} and $\sum_m\alpha_m=1$ give
\[
F(z)\le F_\eps(z)\le F(z)+\eps/4,
\qquad F_\eps^\star\le F^\star+\eps/4.
\]
Therefore
\[
F(\bar z^K)-F^\star
\le F_\eps(\bar z^K)-F_\eps^\star+\eps/4,
\]
which proves Proposition~\ref{prop:shared-model}. The approximation term is $\eps/4$, not $N\eps/4$, because each local empirical loss is averaged and the site weights sum to one.

\subsection{Proof of Proposition~\ref{prop:strong-convex}}
Suppose $f_{1,\eps}$ is $\mu_c$-strongly convex. In Lemma~\ref{lem:one-step}, strong convexity contributes the additional term
$-\mu_c\|\beta^k-\beta^\star\|^2/2$ to the right-hand side. Set
$a_k=\|\beta^k-\beta^\star\|^2$ and choose
$\eta_{k+1}=\{2L_\eps+(k+1)\mu_c\}^{-1}$. Then
\begin{align*}
\frac{a_k-a_{k+1}}{2\eta_{k+1}}-\frac{\mu_c}{2}a_k
&=\left(L_\eps+\frac{k\mu_c}{2}\right)a_k
 -\left(L_\eps+\frac{(k+1)\mu_c}{2}\right)a_{k+1}.
\end{align*}
Hence these first-block distance terms telescope in the unweighted sum. The multiplier and $A_2z$ distance differences also telescope without weighting. Taking expectations removes the martingale term and Assumption~\ref{ass:oracle} gives
\[
\sum_{k=0}^{K-1}\mathbb E\!\left[
 f_{1,\eps}(\beta^{k+1})+P(z^{k+1})-F_\eps^\star
 +\langle\widetilde\lambda,r^{k+1}\rangle\right]
\le C_{\rm sc}+\sigma_b^2\sum_{k=0}^{K-1}\eta_{k+1},
\]
where $C_{\rm sc}<\infty$ depends only on the initial primal-dual distances, $L_\eps$, $\varrho$, and the norm of the test multiplier, but not on $K$. Convexity transfers the left-hand side to the unweighted averages $(\widehat\beta^K,\widehat z^K)$. As in the proof of Theorem~\ref{thm:convex-rate}, after a sample path is realized choose $\widetilde\lambda=q\widehat r^K/\|\widehat r^K\|$ when $\widehat r^K\ne0$ and zero otherwise, where $\widehat r^K=A_1\widehat\beta^K+A_2\widehat z^K$. Since
$\sum_{k=0}^{K-1}\eta_{k+1}=O(\log K)$, division by $K$ yields the stated $O(\log K/K)$ joint objective-feasibility bound. If $q>\|\lambda^\star\|$, the saddle-point inequality again gives feasibility at the same order. This proves Proposition~\ref{prop:strong-convex}.

\section{Conditional weakly convex extension}\label{app:weakly-convex}
\subsection{Coordinator subproblem}
For $t\ge0$, with $a>1$ for MCP and $a>2$ for SCAD, the penalties are
\[
p_{\lambda,a}^{\rm MCP}(t)=
\begin{cases}
\lambda t-t^2/(2a),&0\le t\le a\lambda,\\
a\lambda^2/2,&t>a\lambda,
\end{cases}
\]
and
\[
p_{\lambda,a}^{\rm SCAD}(t)=
\begin{cases}
\lambda t,&0\le t\le\lambda,\\
\dfrac{-t^2+2a\lambda t-\lambda^2}{2(a-1)},&\lambda<t\le a\lambda,\\[1mm]
\dfrac{(a+1)\lambda^2}{2},&t>a\lambda.
\end{cases}
\]
MCP is $1/a$-weakly convex and SCAD is $1/(a-1)$-weakly convex. If $P_{\rm nc}$ is $\kappa$-weakly convex, then
\begin{align*}
P_{\rm nc}(z)+\frac{M\varrho}{2}\|z-v\|^2
={}\left(P_{\rm nc}(z)+\frac\kappa2\|z\|^2\right)
+\frac{M\varrho-\kappa}{2}\|z\|^2
-M\varrho\langle z,v\rangle+\text{constant}.
\end{align*}
Thus $M\varrho>\kappa$ makes the coordinator objective strongly convex and proves Proposition~\ref{prop:coordinator-update}. It proves well-posedness only; it does not restore convexity of the global objective.

\subsection{Proof of the conditional stationarity result}
Let $w^k=(\beta^k,z^k,\lambda^k)$ and let $e^k$ be the gradient error. Summing \eqref{eq:ncdescent} under $\sum_k\|e^k\|^2<\infty$ gives
\begin{align*}
c_1\sum_{k=0}^\infty\|w^{k+1}-w^k\|^2
\le\Phi(w^0)-\inf\Phi+c_2\sum_{k=0}^\infty\|e^k\|^2<\infty.
\end{align*}
Consequently $w^{k+1}-w^k\to0$ and $e^k\to0$. The dual update gives
\[
A_1\beta^{k+1}+A_2z^{k+1}=\varrho^{-1}(\lambda^{k+1}-\lambda^k)\to0.
\]
Equation~\eqref{eq:ncrelative} yields $\dist(0,\partial\Phi(w^{k+1}))\to0$. Boundedness provides convergent subsequences, and closedness of the limiting-subdifferential graph transfers criticality to every cluster point. The assumed identification of critical points of $\Phi$ with the KKT set gives the three KKT conditions in Theorem~\ref{thm:conditional-stationarity}.

If $\Phi$ has the Kurdyka--{\L}ojasiewicz property and $\sum_k\|e^k\|<\infty$, the standard finite-length argument yields $\sum_k\|w^{k+1}-w^k\|<\infty$ and convergence of the whole sequence. The boundedness assumption is essential: MCP and SCAD become flat for large coefficients, so boundedness may require a bounded-level-set condition, compact constraint, full column rank, or small ridge stabilization.

\subsection{Sufficient error regimes and limitations}
\begin{lemma}[Error regimes]\label{lem:error-regimes}
Let $e^k$ be the difference between the stochastic estimator and the full stacked gradient.
\begin{enumerate}
\item If all workers use full local gradients from some finite iteration onward, then $e^k=0$ eventually and both $\sum_k\|e^k\|^2$ and $\sum_k\|e^k\|$ are finite.
\item If deterministic errors satisfy $\|e^k\|\le c(k+1)^{-1-\delta}$ for some $\delta>0$, then the absolute-summability condition holds.
\item If $\mathbb{E}_k\|e^k\|^2\le C/b_k$ and $\sum_k b_k^{-1}<\infty$, then $\sum_k\|e^k\|^2<\infty$ almost surely. This does not imply $\sum_k\|e^k\|<\infty$.
\item A variance-reduced method may satisfy
\begin{align*}
\mathbb{E}_k\|e^{k+1}\|^2
\le a\|e^k\|^2+c\|\beta^{k+1}-\beta^k\|^2,
0<a<1.
\end{align*}
Then an estimator-memory term can be added to the Lyapunov function, but the constants and update rules must be rederived for the chosen estimator.
\end{enumerate}
\end{lemma}
\begin{proof}
The first two statements follow from finite support and comparison with a convergent $p$-series. For the third, Tonelli's theorem gives
\[
\mathbb{E}\sum_{k=0}^\infty\|e^k\|^2
=\sum_{k=0}^\infty\mathbb{E}\|e^k\|^2
\le C\sum_{k=0}^\infty b_k^{-1}<\infty.
\]
A nonnegative random variable with finite expectation is finite almost surely. The fourth statement is the standard contraction mechanism behind memory-based variance reduction: a weighted estimator-error term can be added to the descent function and absorbed under suitable parameters.
\end{proof}
The increasing-batch experiment should therefore be interpreted cautiously. Because the batch size is capped by $n_m$, the estimator eventually becomes exact if the run is sufficiently long. Before that point, the experiment illustrates a decreasing stationarity residual but does not independently verify the almost-sure assumptions. A fixed batch is expected to stabilize in a noise-determined neighborhood.

\end{document}